\documentclass[a4paper,11pt]{article}

\usepackage{fullpage}

\usepackage{amsmath}
\usepackage{amsthm}
\usepackage{amsfonts}
\usepackage{amssymb}

\usepackage{xspace}

\newtheorem{lemma}{Lemma}
\newtheorem{theorem}{Theorem}

\newtheorem{proposition}{Proposition}
\theoremstyle{definition}
\newtheorem{definition}{Definition}

\newcommand{\h}{\mathcal{H}}
\newcommand{\U}{\mathcal{U}}
\newcommand{\F}{\mathcal{F}}
\newcommand{\N}{\mathbb{N}}

\newcommand{\Q}{\mathcal{Q}}

\newcommand{\NP}{$\mathcal{NP}$\xspace}
\renewcommand{\P}{$\mathcal{P}$\xspace}

\newcommand{\phcollapse}{\NP$\subseteq$ co-\NP/poly\xspace}

\newcommand{\MOS}[1]{Min Ones SAT($#1$)\xspace}

\newcommand{\FOO}[1]{\mathcal{Z}(#1)\xspace}

\newcommand{\zeroclosure}{\Delta}
\newcommand{\sunflowerrestriction}{\nabla}
\newcommand{\nonzeroclosedcore}{\Pi}

\newcommand{\mergeable}{mergeable\xspace}
\newcommand{\mergeability}{mergeability\xspace}

\newcommand{\Mergeability}{Mergeability\xspace}
\newcommand{\MergeAbility}{Mergeability\xspace}
\newcommand{\mergeoperation}{merge operation\xspace}

\title{Preprocessing of Min Ones Problems: A Dichotomy}
\author{Stefan Kratsch \and Magnus Wahlstr\"om}

\begin{document}

\maketitle

\begin{abstract}
A parameterized problem consists of a classical problem and an additional component, the so-called \emph{parameter}. This point of view allows a formal definition of preprocessing: Given a parameterized instance~$(I,k)$, a \emph{polynomial kernelization} computes an equivalent instance~$(I',k')$ of size and parameter bounded by a polynomial in~$k$.
We give a complete classification of Min Ones Constraint Satisfaction problems, i.e., \MOS{\Gamma}, with respect to admitting or not admitting a polynomial kernelization (unless \phcollapse). For this we introduce the notion of \mergeability. If all relations of the constraint language~$\Gamma$ are \mergeable, then a new variant of sunflower kernelization applies, based on non-zero-closed cores.
We obtain a kernel with~$O(k^{d+1})$ variables and polynomial total size, where~$d$ is the maximum arity of a constraint in~$\Gamma$, 
comparing nicely with the bound of~$O(k^{d-1})$ vertices for the less general and
arguably simpler~$d$-\textsc{Hitting Set} problem.  Otherwise, any relation in~$\Gamma$ that is not \mergeable permits us to construct a log-cost selection formula, i.e., an~$n$-ary selection formula with~$O(\log n)$ true local variables.
From this we can construct our lower bound using recent results by Bodlaender et al.\ as well as Fortnow and Santhanam, proving that there is no polynomial kernelization, unless \phcollapse and the polynomial hierarchy collapses to the third level.
\end{abstract}

\section{Introduction}\label{section:introduction}

Preprocessing and data reduction are ubiquitous, especially in the context of combinatorially hard problems. Of course, it is a commonplace that there can be no polynomial-time algorithm that provably shrinks every instance of an \NP-hard problem, unless~\P= \NP.
Still, 
there does in fact exist a formal notion of efficient preprocessing, coming from the field of parameterized complexity. There, problems are considered with an additional component, the so-called \emph{parameter}, intended to express the difficulty of a problem instance, e.g., solution size, nesting depth, or treewidth. This way preprocessing can be defined as a polynomial-time mapping~$K:(I,k)\mapsto(I',k')$ such that~$(I,k)$ and~$(I',k')$ are equivalent and~$k'$ as well as the size of~$I'$ are bounded by a polynomial in the parameter~$k$;~$K$ is called a \emph{polynomial kernelization}.
Parameterized complexity originated as a multivariate analysis of algorithms, motivated by the huge difference in (often trivial)~$n^{f(k)}$ versus~$f(k)n^c$ algorithms, the latter having a much better scalability. Kernelization is one possible technique to prove \emph{fixed-parameter tractability} (i.e., the existence of an~$f(k)n^c$ algorithm). Indeed it is known that a problem is fixed-parameter tractable if and only if it admits a kernelization (see~\cite{DF1998}). However, this relation does not imply kernelizations with a polynomial size bound; achieving the strongest size bounds or at least breaking the polynomial barrier is of high interest. Consider for example the list of improvements for \textsc{Feedback Vertex Set}, from the first polynomial kernel~\cite{BurrageEFLMR06}, to cubic~\cite{CubicFVS}, and now quadratic~\cite{QuadraticFVS}; the existence of a linear kernel is still an open problem.
Recently a seminal paper by Bodlaender, Downey, Fellows, and Hermelin~\cite{BodlaenderDFH08} provided the first polynomial lower bounds on the kernelizability of some problems, based on hypotheses in classical complexity. Using results by Fortnow and Santhanam~\cite{FS08}, they showed that so-called \emph{compositional} parameterized problems admit no polynomial kernelizations unless \phcollapse; by Yap~\cite{Yap83}, this would imply that the polynomial hierarchy collapses.  The existence of such lower bounds has sparked high activity in the field (see related work below).

Constraint satisfaction problems (CSP) are a fundamental and general problem setting, encompassing a wide range of natural problems, e.g., satisfiability, graph modification, and covering problems. CSPs are posed as restrictions, called \emph{constraints}, on the feasible assignments to a set of variables. The constraints are applications of relations from a given constraint language~$\Gamma$ to tuples of variables.
The complexity of deciding feasibility of a CSP or finding an assignment that optimizes a certain goal varies according to the constraint language. E.g., consider \textsc{Clique} as a Max Ones SAT($\{\neg x\vee\neg y\}$) problem, which is hard to approximate and also~W[1]-complete when parameterized by the size of the clique.
Khanna et al.~\cite{KhannaSTW00} classified Boolean CSPs according to their approximability, for the questions of optimizing either the weight of a solution (Min/Max Ones SAT problems) or the number of satisfied or unsatisfied constraints (Min/Max SAT problems).
We study the kernelization properties of \MOS{\Gamma}, parameterized by the number of true variables, and classify these problems into admitting or not admitting a polynomial kernelization. We point out that Max SAT($\Gamma$), as a subset of Max SNP~(cf.~\cite{KhannaSTW00}), admits polynomial kernelizations independent of~$\Gamma$~\cite{K09}.

\vskip4pt
\noindent\textbf{Related work}\hskip8pt In the literature there exists an impressive list of problems that admit polynomial kernels (in fact often linear or quadratic); giving stronger and stronger kernels has become its own field of interest. We name only a few results for problems that also have a notion of arity:~$O(k^{d-1})$ universe size for \textsc{Hitting Set} with set size at most~$d$~\cite{Abu-Khzam07},~$O(k^{d-1})$ vertices for packing~$k$ vertex disjoint copies of a~$d$-vertex graph~\cite{Moser09}, and~$O(k^d)$ respectively~$O(k^{d+1})$ base set size for any problem from MIN F$^+\Pi_1$ or MAX NP with at most~$d$ variables per clause~\cite{K09}.

Let us also mention a few lower bound results that are based on the framework of Bodlaender et al.~\cite{BodlaenderDFH08}.
First of all, Bodlaender et al.~\cite{DBLP:conf/esa/BodlaenderTY09} provided kernelization-preserving reductions, which can be used to extend the applicability of the lower bounds.
Using this, Dom et al.~\cite{colorsids} gave polynomial lower bounds for a number of problems, among them \textsc{Steiner Tree} and \textsc{Connected Vertex Cover}. Furthermore they considered problems that have a~$k^{f(d)}$ kernel, where~$k$ is the solution size and~$d$ is a secondary parameter (e.g., maximum set size), and showed that there is no kernel with size polynomial in~$k+d$; for, e.g., \textsc{Hitting Set}, \textsc{Set Cover}, and \textsc{Unique Coverage}.
Fernau et al.~\cite{DBLP:conf/stacs/FernauFLRSV09} showed that \textsc{Leaf Out Branching} does not admit a polynomial kernelization, while \textsc{Rooted Leaf Out Branching} does. They express that this gives a Turing kernelization for \textsc{Leaf Out Branching}, by creating one kernel for each choice of the root.
In~\cite{hfreenokernel} the present authors show that a certain Min Ones CSP problem does not admit a polynomial kernel and employ this bound to show that there are~$\h$-free edge deletion respectively edge editing problems that do not admit a polynomial kernel.

\vskip4pt
\noindent\textbf{Our work}\hskip8pt We give a complete classification of \MOS{\Gamma} problems with respect to admittance of polynomial kernelizations. Apart from the hardness dichotomy due to Khanna et al.~\cite{KhannaSTW00}, we distinguish constraint languages~$\Gamma$ by being \emph{\mergeable} or containing at least one relation that is not \mergeable.
For the first case, we provide a new polynomial kernelization based on \emph{non-zero-closed cores}. For the latter we show that \MOS{\Gamma} is either polynomial-time solvable or does not admit a polynomial kernelization, unless \phcollapse.

\vskip4pt
\noindent\textbf{Structure of the paper}\hskip8pt We introduce some basic notation and the notion of \mergeability in Sections~\ref{section:preliminaries} and~\ref{section:decomposability}. 
Sections~\ref{section:upper} and~\ref{section:lower} then form the main part of this work, i.e., the general polynomial kernelization for \MOS{\Gamma} when all relations of~$\Gamma$ are \mergeable, and the lower bound for constraint languages that contain at least one relation that is not \mergeable. We conclude in Section~\ref{section:conclusion}, with a discussion of implications as well as open problems.

\section{Boolean Constraint Satisfaction Problems}\label{section:preliminaries}

A \emph{constraint} is an application of a relation~$R$ to a tuple of variables~$(x_1, \ldots, x_r)$, requiring that~$R(x_1, \ldots, x_r)$ holds, allowing repeated variables (e.g.,~$R(x,x,y)$).  A \emph{constraint language} is a set~$\Gamma$ of relations;
we shall require throughout that every constraint language~$\Gamma$ is finite, and contains only relations over the boolean domain.
A formula~$\F$ over~$\Gamma$ is a conjunction of constraints using relations~$R \in \Gamma$,
and~$V(\F)$ denotes the set of variables that occur in~$\F$. 
An assignment to the variables of~$\F$ satisfies~$\F$ if every constraint in~$\F$ holds under the assignment. The \emph{weight} of an assignment is the number of variables that it sets to true.
Fixing a finite set~$\Gamma$ with relations over the boolean domain defines a \MOS{\Gamma} problem:
\begin{quotation}
\noindent\textbf{Input:} A formula~$\F$ over a finite constraint language~$\Gamma$; an integer~$k$.

\noindent\textbf{Parameter:} $k$.

\noindent\textbf{Task:} Decide whether there is a satisfying assignment for~$\F$ of weight at most~$k$.
\end{quotation}
As an example, if~$R(x,y)=\{(0,1),(1,0),(1,1)\}$, then \MOS{R} is the well-known problem \textsc{Vertex Cover}.
The approximation properties of such problems have been classified by Khanna et al.~\cite{KhannaSTW00}; in particular, we have the following.
\begin{theorem}[\cite{KhannaSTW00}] \label{th:mosnp}
Let~$\Gamma$ be a finite set of relations over the boolean domain.  If~$\Gamma$ is zero-valid, Horn, or width-2 affine (i.e., implementable by assignments,~$(x=y)$, and~$(x \neq y)$), then \MOS{\Gamma} is in \P; otherwise it is \NP-complete.
\end{theorem}

SAT($\Gamma$) denotes the problem of deciding whether any satisfying assignment exists; the classical complexity of these problems was classified by Schaefer~\cite{schaefersat}, and the parameterized complexity, for the question of finding a satisfying assignment with \emph{exactly}~$k$ true variables, has been classified by Marx \cite{marxcsp}.  The problem \MOS{\Gamma} is fixed-parameter tractable for every finite~$\Gamma$, by a simple branching algorithm; see~\cite{marxcsp}.

We need to define a number of types of constraints.
Let~$\Gamma$ be a finite set of relations over the boolean domain. 
We say that~$\Gamma$ \emph{implements} a relation~$R$ if~$R$ is the set of satisfying assignments for a formula over~$\Gamma$,
i.e.,~$R(x_1,\ldots,x_r) \equiv \bigwedge_i R_i(x_{i1}, \dots, x_{it})$ where each~$R_i \in \Gamma$ (we do not automatically allow the equality relation unless~$= \in \Gamma$).  
A \emph{positive clause} is a disjunction of (non-negated) variables. A \emph{negative clause} is a disjunction of negated variables.
We say that a constraint is \emph{zero-valid} if a tuple of zeros satisfies it. A constraint is \emph{Horn} if it can be implemented by disjunctions containing at most one unnegated variable each, \emph{dual Horn} if it can be implemented by disjunctions containing at most one negated variable each, and \emph{IHSB-} (Implicative Hitting Set Bounded-) if it can be implemented by assignments, implications, and negative clauses.
These constraint types can also be characterized by closure properties.  For two tuples~$\alpha=(\alpha_1,\dots,\alpha_r)$,~$\beta=(\beta_1,\dots,\beta_r)$, let~$\alpha \land \beta = (\alpha_1 \land \beta_1, \dots, \alpha_r \land \beta_r)$, and likewise for~$\alpha \lor \beta$, and write~$\alpha \leq \beta$ if~$\alpha_i \leq \beta_i$  for every~$1 \leq i \leq r$ (where~$0$ and~$1$ are used for false and true values, respectively). 
We then have that a constraint~$R$ is Horn if and only if it is closed under intersection, i.e., if~$\alpha, \beta \in R$, then $\alpha \land \beta \in R$,
and a constraint is dual Horn if and only if it is closed under disjunction.  
Likewise, a constraint~$R$ is IHSB- if and only if it is closed under an operation~$\alpha \land (\beta \lor \gamma)$ for tuples~$\alpha, \beta, \gamma$ in~$R$.  See~\cite{CreignouV08} for more on this.
A constraint language~$\Gamma$ is zero-valid (one-valid, Horn, dual Horn, IHSB-) if every~$R \in \Gamma$ is.

\section{\MergeAbility}\label{section:decomposability}
\label{sec:polymorphism}
The characterization of the dichotomy of the kernelizability of \MOS{\Gamma} given in this paper centers around a newly introduced property we refer to as \emph{\mergeability.} 
Specifically, we will see that for any finite set~$\Gamma$ of relations over the boolean domain, \MOS{\Gamma} admits a polynomial kernelization if either \MOS{\Gamma} is in \P or every relation~$R \in \Gamma$ is \mergeable; in every other case, \MOS{\Gamma} admits no polynomial kernelization unless \phcollapse (which would imply that the polynomial hierarchy collapses to the third level).  
In this section, we define this property and give some basic results about it.

\begin{definition} \label{def:sunflowerdecomp}
Let~$R$ be a relation on the boolean domain.  Given four (not necessarily distinct) tuples~$\alpha,\beta,\gamma,\delta \in R$, we say that the \emph{\mergeoperation} applies if~$\alpha \land \delta \leq \beta \leq \alpha$ and~$\beta \land \gamma  \leq \delta \leq \gamma$.  
If so, then \emph{applying} the \mergeoperation produces the tuple~$\alpha \land (\beta \lor \gamma)$.
We say that~$R$ is \emph{\mergeable} if for any four tuples~$\alpha,\beta,\gamma,\delta \in R$ for which the \mergeoperation applies, 
we have~$\alpha \land (\beta \lor \gamma) \in R$.  
\end{definition}

We show some basic results about \mergeability.
First, we show an alternate presentation of the property; this perspective will be important in Section~\ref{section:upper}, when sunflowers are introduced. 
\begin{proposition} \label{prop:sunflowerdecomp}
Let~$R$ be a relation of arity~$r$ on the boolean domain.  Partition the positions of~$R$ into two sets,  called the \emph{core} and the \emph{petals}; w.l.o.g. assume that positions~$1$ through~$c$ are the core, and the rest the petals.  Let~$(\alpha_C,\alpha_P)$, where~$\alpha_C$ is a~$c$-ary tuple and~$\alpha_P$ an~$(r-c)$-ary tuple, 
denote the tuple whose first~$c$ positions are given by~$\alpha_C$, and whose subsequent positions are given by~$\alpha_P$.
Consider then the following four tuples.
\begin{eqnarray}
\alpha &=& (\alpha_C, \alpha_P) \nonumber \\
\beta  &=& (\alpha_C, 0) \nonumber \\
\gamma &=& (\gamma_C, \gamma_P) \nonumber \\
\delta &=& (\gamma_C, 0) \nonumber
\end{eqnarray}
If~$\alpha$ through~$\delta$ are in~$R$, then the \mergeoperation applies, giving us
\[
(\alpha_C, \alpha_P \land \gamma_P) \in R.
\]
Furthermore, for any four tuples 
to which the \mergeoperation applies, there is a partitioning of the positions into core and petals such that the tuples can be written in the above form.
\end{proposition}

It is straight-forward that this property is preserved by implementations.

\begin{proposition} \label{prop:preserved}
\Mergeability is preserved by assignment and identification of variables, i.e., 
if~$R$ is \mergeable, then so is any relation produced from~$R$ by these operations.
Further, any relation implementable by \mergeable relations is \mergeable.
\end{proposition}

Next, we show what \mergeability implies for a zero-valid relation.

\begin{lemma} \label{lm:zeroprop}
Any zero-valid relation~$R$ which is \mergeable is also IHSB-, 
and can therefore be implemented using negative clauses and implications.
\end{lemma}
\begin{proof}
We show that~$\alpha \land (\beta \lor \gamma) \in R$ for all tuples~$\alpha,\beta,\gamma \in R$.
First, for any two tuples~$\alpha,\beta \in R$, we can apply the \mergeoperation to
the tuples~$\alpha$,~$0$,~$\beta$,~$0$, to show that~$\alpha \land \beta \in R$.  
It can then be checked that the operation applies to the tuples~$\alpha$,~$(\alpha \land \beta)$,~$(\alpha \land \gamma)$, and~$(\alpha  \land \beta \land \gamma)$, and that this implies~$\alpha \land (\beta \lor \gamma) \in R$.
As previously mentioned, this shows that~$R$ is IHSB-.
Note that assignments add no expressive power when~$R$ is zero-valid.
\end{proof}

Note that by Proposition~\ref{prop:preserved}, this shows that the only zero-valid relations that can be produced from a \mergeable relation by assigning or identifying variables are IHSB-.
However, this still leaves room for other positive examples; e.g., the constraints~$(x+y+z=1 \textrm{ (mod 2)})$ and~$((x=y) \rightarrow z)$ are both \mergeable.

\section{Kernelization} \label{section:upper}

In this section, we show that \MOS{\Gamma} admits a polynomial kernelization if all relations in~$\Gamma$ are \mergeable.
For the purpose of describing our kernelization we first define a sunflower of tuples, similarly to the original sunflower definition for sets.  We point out that a similar though more restricted definition for sunflowers of tuples was given by Marx~\cite{marxcsp}; accordingly the bounds of our sunflower lemma are considerably smaller.

\begin{definition}Let~$\U$ be a finite set, let~$d\in\N$, and let~$\h\subseteq \U^d$. A \emph{sunflower} (\emph{of tuples}) \emph{with cardinality~$t$ and core~$C\subseteq\{1,\dots,d\}$} in~$\U$ is a subset consisting of~$t$ tuples that have the same element at all positions in~$C$ and, in the remaining positions, no element occurs in more than one tuple. The set of remaining positions~$P=\{1,\dots,d\}\setminus C$ is called the \emph{petals}.
\end{definition}

As an example,~$(x_1, \dots, x_c, y_{11}, \dots, y_{1p})$, $\dots$, $(x_1, \dots, x_c, y_{t1}, \dots, y_{tp})$ is a sunflower of cardinality~$t$ with core~$C=\{1,\dots,c\}$, if all~$y_{ij}$ and~$y_{i'j'}$ are distinct when~$i \neq i'$. Note that, differing from Marx~\cite{marxcsp} variables in the petal positions may also occur in the core.
For sets of tuples~$\h\subseteq \U^d$, we give a variant of Erd\H{o}s' and Rado's Sunflower Lemma~\cite{Erdos1960}. The proof is along the same lines as the original, only requiring an additional factor of~$d!$ for picking the shared core positions. Same as the Sunflower Lemma, this immediately gives a polynomial-time algorithm for finding a sunflower of tuples.

\begin{lemma}\label{lemma:sunflower}
Let~$\U$ be a finite set, let~$d\in\N$, and let~$\h \subseteq \U^d$. If the size of~$\h$ is greater than~$k^d(d!)^2$, then it contains a sunflower of cardinality~$k+1$.
\end{lemma}

\begin{proof}
If~$d=1$, then a sunflower of size~$k+1$ can be easily found, since any~$k+1$ tuples of arity~$d=1$ form a sunflower with empty core. Now for induction, assume the lemma to be true for all~$d'\leq d-1$.

Let~$X$ contain~$\{x_1,\dots,x_d\}$ for each tuple~$(x_1,\dots,x_d)\in\h$. Select a maximal pairwise disjoint subset~$F\subseteq X$. If~$|F|\geq k+1$ then its elements correspond to~$k+1$ tuples that share no variable, i.e., a sunflower with empty core. Otherwise, if~$|F|\leq k$, then all other sets of~$X$ have a non-empty intersection with some element of~$F$. Since the sets correspond to the tuples of~$\h$ there must be an element of some set in~$F$, say~$x$, that occurs in at least~$|\h|/kd$ tuples of~$\h$, as there are at most~$kd$ such elements. Therefore, there must be a position,~$p\in\{1,\dots,d\}$, such that~$x$ occurs in position~$p$ of at least~$|\h|/kd^2 > k^{d-1}((d-1)!)^2$ tuples of~$\h$. 

Define~$\h'$ by~$\h'=\{(x_1,\dots,x_{p-1},x_{p+1},\dots,x_d)\mid (x_1,\dots,x_{p-1},x,x_{p+1},\dots,x_d)\in \h\}$. Observe that~$\h'\subseteq U^{d-1}$ and~$|\h'|>k^{d-1}((d-1)!)^2$, implying that a sunflower of cardinality~$k+1$ in~$\h'$ can be found in~$\h'$; immediately giving a sunflower in~$\h$.
\end{proof}

Our kernelization requires also the notion of a zero-closed position and the related zero-closure of a relation. The following definition introduces these concepts.

\begin{definition}
Let~$R$ be an~$r$-ary relation.
The relation~$R$ is \emph{zero-closed on position~$i$}, if for every tuple~$(t_1, \dots,t_{i-1}, t_i, t_{i+1}, \dots, t_r) \in R$ we have~$(t_1, \dots, t_{i-1}, 0, t_{i+1}, \dots, t_r)\in R$. A relation is \emph{non-zero-closed} if it has no zero-closed positions. We define functions~$\zeroclosure$,~$\nonzeroclosedcore$, and~$\sunflowerrestriction$:
\begin{itemize}
\item $\zeroclosure_P(R)$ is defined to be the \emph{zero-closure of~$R$ on positions~$P\subseteq\{1,\dots,r\}$}, i.e., the smallest superset of~$R$ that is zero-closed on all positions~$i \in P$.
\item $\nonzeroclosedcore(R)$ denotes the \emph{non-zero-closed core}, i.e., the projection of~$R$ onto all positions that are not zero-closed or, equivalently, the relation on the non-zero-closed positions obtained by forcing~$x_i=0$ for all zero-closed positions.
\item $\sunflowerrestriction_C(R)$ denotes the \emph{sunflower restriction of~$R$ with core~$C$}: w.l.o.g.\ the relation expressed by~$[\sunflowerrestriction_C(R)](x_1, \dots, x_r) = R(x_1, \dots, x_r) \land R(x_1, \dots, x_c, 0, \dots, 0)$ for~$C=\{1,\dots,c\}$. The corresponding \emph{core relation} is the~$|C|$-ary relation given by~$R(x_1, \dots, x_c, 0, \dots, 0)$.
\end{itemize}
The mappings~$\zeroclosure_P$ and~$\sunflowerrestriction_C$ extend also to constraints:~$\sunflowerrestriction_C(R(x_1,\dots,x_r))=[\sunflowerrestriction_C(R)](x_1,\dots,x_r)$. Similarly for~$\nonzeroclosedcore$, but variables in zero-closed positions are removed, e.g., when~$i$ is the only zero-closed position of~$R$ then~$\nonzeroclosedcore(R(x_1,\dots,x_r))=[\nonzeroclosedcore(R)](x_1,\dots,x_{i-1},x_{i+1},\dots,x_r)$.
\end{definition}

\begin{lemma} \label{lemma:sunflimpl}
Let~$R$ be a \mergeable relation and let~$C\cup P$ be a partition of its positions into core and petals. There is an implementation of~$\sunflowerrestriction_C(R)$ using~$\zeroclosure_P(\sunflowerrestriction_C(R))$ and implications.
\end{lemma}
\begin{proof}
By Prop.~\ref{prop:preserved},~$\sunflowerrestriction_C(R)$ must be \mergeable.
Further, assigning any set of values to the variables in the core produces a zero-valid relation on the petals, which is still \mergeable.
Thus by Lemma~\ref{lm:zeroprop}, this relation on the petals has an implementation using negative clauses and implications.
For any tuple~$\sigma$ in~$\sunflowerrestriction_C(R)$ or~$\zeroclosure_P(\sunflowerrestriction_C(R))$, let its \emph{core assignment} be the values it assigns to the positions in~$C$. By definition, we have~$\sunflowerrestriction_C(R)\subseteq\zeroclosure_P(\sunflowerrestriction_C(R))$.

Consider now a tuple~$\sigma \in \zeroclosure_P(\sunflowerrestriction_C(R))\setminus \sunflowerrestriction_C(R)$.
Assume that~$\sigma$ makes core assignment~$\alpha_C$; thus there is a matching~$\alpha \in \sunflowerrestriction_C(R)$,~$\alpha>\sigma$, with an identical core assignment.  As in Prop.~\ref{prop:sunflowerdecomp}, write~$\alpha=(\alpha_C,\alpha_P)$, and 
consider the constraint on the petals that is formed by core assignment~$\alpha_C$.
We see that this constraint must entail some implication~$(y_i \rightarrow y_j)$, where~$\sigma$ assigns~$y_i=1,y_j=0$ while~$\alpha_P$ assigns~$y_i=y_j=1$.
Further, since~$\sunflowerrestriction_C(R)$ is a sunflower restriction, we have~$\beta=(\alpha_C,0) \in \sunflowerrestriction_C(R)$.
Now if~$(y_i \rightarrow y_j)$ does not hold in~$\sunflowerrestriction_C(R)$ in general, then let~$\gamma=(\gamma_C,\gamma_P) \in \sunflowerrestriction_C(R)$ be a tuple which assigns~$y_i=1$,~$y_j=0$,
and let~$\delta=(\gamma_C,0) \in \sunflowerrestriction_C(R)$.  
By Prop.~\ref{prop:sunflowerdecomp}, we can now apply the \mergeoperation to tuples~$\alpha$ through~$\delta$, showing~$(\alpha_C, \alpha_P \land \gamma_P) \in \sunflowerrestriction_C(R)$.  But this is a tuple with core assignment~$\alpha_C$ which assigns~$y_i=1$,~$y_j=0$, which is a contradiction. 
Thus~$(y_i \rightarrow y_j)$ holds in~$\sunflowerrestriction_C(R)$ regardless of core assignment, and the constraint~$(y_i \rightarrow y_j)$ can be added to our implementation of~$\sunflowerrestriction_C(R)$, removing the tuple~$\sigma$.

Adding all implications between petals which hold in~$\sunflowerrestriction_C(R)$ removes from our implementation all tuples which are in~$\zeroclosure_P(\sunflowerrestriction_C(R))$ but not in~$\sunflowerrestriction_C(R)$, so that the conjunction of~$\zeroclosure_P(\sunflowerrestriction_C(R))$ with all valid implications is an implementation of~$\sunflowerrestriction_C(R)$.
\end{proof}

The following technical lemma proves that the relation~$\zeroclosure_P(\sunflowerrestriction_C(R))$, required by Lemma~\ref{lemma:sunflimpl}, is \mergeable.

\begin{lemma} \label{lemma:zeroclosure}
Let~$R$ be a \mergeable relation and let~$C\cup P$ be a partition of its positions into core and petals. Then~$\zeroclosure_P(\sunflowerrestriction_C(R))$ is \mergeable.
\end{lemma}

\begin{proof}
Recall that~$\zeroclosure_P(\sunflowerrestriction_C(R))$ is the zero-closure on the petal positions of the sunflower restriction of~$R$ with core~$C$.
Let~$R':=\zeroclosure_P(\sunflowerrestriction_C(R))$; assume by way of contradiction that~$R'$ is not \mergeable. Then there are four tuples in~$R'$ such that applying the \mergeoperation
on the tuples creates a tuple not in~$R'$.  Let~$C' \cup P'$ be the partition of the positions of~$R'$ into core and petals that is used in this counterexample.
Grouping the positions of~$R'$ in four groups, written in the order~$(C' \cap C, C' \cap P, P' \cap C, P' \cap P)$,
naming the groups~$W$ through~$Z$, the counterexample can be written as follows.
\begin{eqnarray}
(W_1,X_1,Y_1,Z_1) &\in& R' \label{eqnl1}\\
(W_1,X_1,0,0)     &\in& R' \label{eqnl2}\\
(W_2,X_2,Y_2,Z_2) &\in& R' \label{eqnl3}\\
(W_2,X_2,0,0)     &\in& R' \label{eqnl4}\\
(W_1,X_1, Y_1 \land Y_2, Z_1 \land Z_2) &\notin& R' \label{eqnlconc}
\end{eqnarray}
We will derive a contradiction.  First, we note that for each equation~(\ref{eqnl1})--(\ref{eqnl4}), there is a corresponding tuple in~$\sunflowerrestriction_C(R)$.
\begin{eqnarray}
(W_1,X_{1a},Y_1,Z_{1a}) &\in& \sunflowerrestriction_C(R) \label{eqnlp1}\\
(W_1,X_{1b},0,Z_b)     &\in& \sunflowerrestriction_C(R) \label{eqnlp2}\\
(W_2,X_{2a},Y_2,Z_{2a}) &\in& \sunflowerrestriction_C(R) \label{eqnlp3}\\
(W_2,X_{2b},0,Z_c)     &\in& \sunflowerrestriction_C(R) \label{eqnlp4}
\end{eqnarray}
Here,~$X_{1a}$ and $X_{1b}$ are supersets of~$X_1$, and likewise for~$X_2$,~$Z_1$, and~$Z_2$.~$Z_b$ and~$Z_c$ are arbitrary.
Using that~$\sunflowerrestriction_C(R)$ is a sunflower restriction and \mergeable, we can conclude the following.
\begin{eqnarray}
(W_1,0,0,0)     &\in& \sunflowerrestriction_C(R) \label{eqnl5}\\
(W_2,0,0,0)     &\in& \sunflowerrestriction_C(R) \label{eqnl6}\\
(W_1,X_{1a} \land X_{2a}, Y_1 \land Y_2, Z_{1a} \land Z_{2a}) &\in& \sunflowerrestriction_C(R) \label{eqnl7} 
\end{eqnarray}
The first two come from~(\ref{eqnlp2}) and~(\ref{eqnlp4}); the third is produced by a \mergeoperation on~(\ref{eqnlp1}) and~(\ref{eqnlp3}) using these two.
Now, the tuples which match~$W_1$ on the $W$-variables form a zero-valid relation.  By Lemma~\ref{lm:zeroprop}, this relation is closed under an operation~$(\alpha \land (\beta \lor \gamma))$.
Applying this on the tuples of equations~(\ref{eqnlp1}),~(\ref{eqnlp2}), and~(\ref{eqnl7}) gives us the following conclusion.
\begin{equation}
(W_1, X_{1a} \land (X_{1b} \lor X_{2a}), Y_1 \land Y_2, Z_{1a} \land (Z_b \lor Z_{2a})) \in \sunflowerrestriction_C(R) \label{eqnl9}
\end{equation}
In particular, this tuple matches~(\ref{eqnlconc}) on the $W$- and $Y$-variables, and is a superset of it on the $X$- and $Z$-variables.
Since~$R'$ is the zero-closure of~$\sunflowerrestriction_C(R)$ on the~$X$- and~$Z$-variables, we have a contradiction. 
\end{proof}

Lemmas~\ref{lemma:sunflimpl} and~\ref{lemma:zeroclosure} are the foundation for a sunflower-based kernelization for \MOS{\Gamma}. They show that the sunflower restriction~$\sunflowerrestriction_C(R)$ of some \mergeable~$R$-constraint can be implemented using its \mergeable zero-closure on the petal positions as well as implications. However,~$\zeroclosure_P(\sunflowerrestriction_C(R))$ is not necessarily contained in~$\Gamma$ and such a replacement does not give any immediate reduction of the instance. Indeed, the arity of~$\zeroclosure_P(\sunflowerrestriction_C(R))$ is the same as that of~$R$.

We address this problem by introducing a new measure of difficulty for formulas, namely the sum of non-zero-closed cores, based on the following definition.

\begin{definition}
Let~$\F$ be a formula and let~$R$ be a relation. We define~$\FOO{\F,R}$ as the set of all tuples~$(x_1,\dots,x_t)$ where~$[\nonzeroclosedcore(R)](x_1,\dots,x_t)$ is the non-zero-closed core of an~$R$-constraint in~$\F$.
\end{definition}

For some relations, sunflower restriction, zero-closure, and the relation itself are the same, for certain selections of core and petal positions. See for example the following \mergeable relation:
\begin{align}
R=&\{(0,0,1,0),(0,1,0,0),(0,1,0,1),(1,0,0,0),(1,0,0,1),(1,1,1,0),(1,1,1,1)\}\nonumber \\
=&\sunflowerrestriction_{\{1,2,3\}}(R) = \zeroclosure_{\{4\}}(\sunflowerrestriction_{\{1,2,3\}}(R)) \nonumber
\end{align}
This is also one of the smallest examples, where a sunflower restriction cannot be expressed using the core relation (i.e., all tuples for the core such that the petal variables can take value~$0$), implications, and negative clauses. Here, the core relation is~$\{(0,0,1),(0,1,0),(1,0,0),(1,1,1)\}$, but no implication or negative clause can exclude the tuple~$(0,0,1,1)$ without also excluding other tuples that do occur in the sunflower restriction.
Thus there are \mergeable relations for which a sunflower-based reduction using Lemma~\ref{lemma:sunflimpl} does not lead to any simplification, even in terms of~$\FOO{\F,R}$. We overcome this difficulty by searching for sunflowers among the tuples of~$\FOO{\F,R}$. Those are leveraged into a replacement of the~$R$-constraints that contributed these tuples. The following theorem shows this approach in detail.

\begin{theorem}\label{theorem:equivalentformula}Let~$\Gamma$ be a \mergeable constraint language with maximum arity~$d$. Let~$\F$ be a formula over~$\Gamma$ and let~$k$ be an integer. In polynomial time one can compute a formula~$\F'$ over a \mergeable constraint language~$\Gamma'\supseteq\Gamma$ with maximum arity~$d$, such that every assignment of weight at most~$k$ satisfies~$\F$ if and only if it satisfies~$\F'$ and, furthermore,~$|\FOO{\F',R}|\in O(k^d)$ for every non-zero-valid relation that occurs in~$\F'$.
\end{theorem}

\begin{proof} We begin constructing~$\F'$, starting from~$\F'=\F$. While~$|\FOO{\F',R}|>k^d(d!)^2$ for any non-zero-valid relation~$R$ in~$\F'$, search for a sunflower of cardinality~$k+1$ in~$\FOO{\F',R}$, according to Lemma~\ref{lemma:sunflower}. Let~$C$ denote the core of the sunflower and apply the following replacement. Remove each~$R$-constraint whose non-zero-closed core matches a tuple of the sunflower, and add its sunflower restriction with core~$C$ using an implementation according to Lemma~\ref{lemma:sunflimpl}.
Repeating this step until~$|\FOO{\F',R}|\leq k^d(d!)^2$ for all non-zero-valid relations~$R$ in~$\F'$ completes the construction.

Now, to prove correctness, let us consider a single replacement. We denote the tuples of the sunflower by~$(x_1,\dots,x_c,y_{i1},\dots,y_{ip})$, with~$i\in\{1,\dots,k+1\}$, i.e., w.l.o.g.\ with core~$C=\{1,\dots,c\}$ and petals~$P=\{c+1,\dots,c+p\}$. Let~$\phi$ be any satisfying assignment of weight at most~$k$ and consider any tuple~$(x_1,\dots,x_c,y_{i1},\dots,y_{ip})$ of the sunflower. There must be a constraint~$R(x_1,\dots,x_c,y_{i1},\dots,y_{ip},z_1,\dots,z_t)$ whose non-zero-closed core matches the tuple, w.l.o.g.\ we take the last positions of~$R$ to be zero-closed, let~$Z$ be those positions. Thus~$\phi$ must satisfy~$R(x_1,\dots,x_c,y_{i1},\dots,y_{ip},0,\dots,0)$, since the~$z_i$ are in zero-closed positions. Observe that, by maximum weight~$k$, the assignment~$\phi$ assigns~$0$ to all variables~$y_{i1},\dots,y_{ip}$ for an~$i\in\{1,\dots,k+1\}$. Thus~$\phi$ satisfies also~$R(x_1,\dots,x_c,0,\dots,0)$.
Hence for any constraint~$R(x_1,\dots,x_c,y_{i1},\dots,y_{ip},z_1,\dots,z_t)$, it satisfies~$\sunflowerrestriction_C(R(x_1,\dots,x_c,y_{i1},\dots,y_{ip},z_1,\dots,z_t))$ too. This permits us to replace each~$R$-constraint, whose non-zero-closed core matches a tuple of the sunflower, by an implementation of its sunflower restriction with core~$C$, according to Lemma~\ref{lemma:sunflimpl}. The implementation uses~$\zeroclosure_{P\cup Z}(\sunflowerrestriction_C(R(x_1,\dots,x_c,y_{i1},\dots,y_{ip},z_1,\dots,z_t)))$ and implications. By Lemma~\ref{lemma:zeroclosure} the added constraints~$\zeroclosure_{P\cup Z}(\sunflowerrestriction_C(R(x_1,\dots,x_c,.,\dots,.)))$ are \mergeable, implying that all constraints in~$\F'$ are \mergeable.

To establish that the construction can be performed efficiently, i.e., in time polynomial in the size of~$\F$, we use as a measure of~$\F'$ the sum of~$|\FOO{\F',R}|$ over all relations~$R$ occurring in~$\F'$. 
First, let us observe that, initially, this measure is bounded by the size of~$\F$ since each~$R$-constraint of~$\F'$ contributes at most one tuple to the corresponding set~$\FOO{\F',R}$ (recall that we start with~$\F'=\F$). Consider again the replacement made in each step: All~$R$-constraints matching one of the tuples of the sunflower are replaced by an implementation using~$\hat{R}=\zeroclosure_{P\cup Z}(\sunflowerrestriction_C(R))$ and implications. It is crucial to
observe that all added constraints contribute the same tuple to~$\FOO{\F',\hat{R}}$, consisting only of variables with positions in~$C$. This is caused by the application of the zero closure~$\zeroclosure$ on all positions but those in~$C$. Hence the~$k+1$ tuples of the sunflower are removed, as all matching~$R$-constraints are replaced, and only one new tuple is added to the set~$\FOO{\F',\hat{R}}$.
This decreases the measure, implying that the modification step is applied at most a number of times polynomial in the size of~$\F$.

Finally let us express the fact that each iteration of the replacement can be done efficiently. The set~$\FOO{\F',R}$ can be generated in one pass over the formula and since the arity is bounded by~$d$ there is only a constant number of relations.
The applications of Lemma~\ref{lemma:sunflower} to find a sunflower among the tuples of the sets~$\FOO{\F',R}$ take time polynomial in~$|\F|$, since the size~$|\FOO{\F',R}|\in O(|\F|)$. Observe that the size of~$\F'$ is bounded by a polynomial in~$|\F|$ at all times, since there is only a polynomial number of possible constraints of arity at most~$d$ on the variables of~$\F$.
\end{proof}

Now we are able to derive a polynomial kernelization for \MOS{\Gamma}. For a given instance~$(\F,k)$, it first generates an equivalent formula~$\F'$ according to Theorem~\ref{theorem:equivalentformula}. However,~$\F'$ will not replace~$\F$, rather, it allows us to remove variables from~$\F$ based on conclusions drawn from~$\F'$. This approach avoids the obstacle of a possible lack of expressibility from using only the language~$\Gamma$, and requires no additional assumptions or annotations to be made.

\begin{theorem} \label{th:upperfinal}
Let~$\Gamma$ be a \mergeable constraint language.  Then \MOS{\Gamma} admits a polynomial kernelization.
\end{theorem}

\begin{proof}
Let~$(\F,k)$ be an instance of \MOS{\Gamma} and let~$d$ be the maximum arity of relations in~$\Gamma$. According to Theorem~\ref{theorem:equivalentformula}, we generate a formula~$\F'$, such that assignments of weight at most~$k$ are satisfying for~$\F$ if and only if they are satisfying for~$\F'$. Moreover, for each non-zero-valid relation~$R$, we have that~$|\FOO{\F',R}|\in O(k^d)$. Note that constraints of~$\F'$ have maximum arity~$d$. We allow the constant~$0$ to be used for replacing variables; a construction for this not using~$(x=0)$ follows at the end of the proof.

First, according to Lemma~\ref{lm:zeroprop}, we replace each zero-valid constraint of~$\F'$ by an implementation through negative clauses and implications. Next, we address variables that occur only in zero-closed positions constraints in~$\F'$. By definition of zero-closed positions it is immediate that setting such a variable to~$0$, does not affect the possible assignments for the other variables. By equivalence of~$\F$ and~$\F'$ with respect to assignments of weight at most~$k$, the same is true for~$\F$. We replace all such variables by the constant~$0$ in~$\F$ and~$\F'$, maintaining the equivalence with respect to assignments of weight at most~$k$.

Now, let~$X$ be the set of variables that occur in a non-zero-closed position of some non-zero-valid constraint of~$\F'$. For each variable~$x\in X$ count the number of variables that are implied by~$x$, i.e., that have to take value~$1$ if~$x=1$, by implication constraints in~$\F'$. If the number of those variables is at least~$k$, then there is no satisfying assignment of weight at most~$k$ for~$\F'$ that assigns~$1$ to~$x$. By equivalence of~$\F$ and~$\F'$ with respect to such assignments, we replace all occurrences of such a variable~$x$ by the constant~$0$, again maintaining the equivalence property.
Finally we replace all variables~$y\in V(\F')\setminus X$, that are not implied by a variable from~$X$ in~$\F'$, by the constant~$0$ in~$\F$ and~$\F'$. Note that such variables~$y$ occur only in zero-closed positions and in implications. It can be easily verified that this does not affect satisfiability with respect to assignments of weight at most~$k$.
For efficiency of this modification consider the fact that the number of implications in~$\F'$ is polynomial in the initial size of~$\F$, since there are at most two implications per pair of variables of~$\F$. This completes the kernelization.

Now we prove a bound of~$O(k^{d+1})$ on the number of variables in~$\F$. First, we observe that all remaining variables of~$\F$ must occur in a non-zero-closed position of some constraint of~$\F'$. We begin by bounding the number of variables that occur in a non-zero-closed position of some non-zero-valid~$R$-constraint, i.e., the remaining variables of the set~$X$. Observe that such a variable must occur in the corresponding tuple of~$\FOO{\F',R}$. Since there is only a constant number of relations of arity at most~$d$ and since~$\FOO{\F',R}\in O(k^d)$, this limits the number of such variables by~$O(k^d)$. For all other variables, their non-zero-closed occurrences must be in implications, since negative clauses are zero-closed on all positions. Thus, these variables must be implied by a variable of~$X$. Since each variable implies at most~$k-1$ other variables, we get an overall bound of~$O(k^{d+1})$.
Finally, the total size of~$\F$ is polynomial for a fix~$d$, since the number of variables is polynomial and the arity of the constraints is bounded.

To express the~$0$-constant, we add~$k+1$ new variables~$z_1,\dots,z_{k+1}$. Every constraint with at least one~$0$ is replaced by~$k+1$ copies, each time replacing~$0$ with a different~$z_i$. Clearly one of the~$z_i$ takes value~$0$ in any assignment of weight at most~$k$. Hence the original constraints with constant~$0$ are enforced. Conversely, given a satisfying assignment of weight at most~$k$ for the formula before making this replacement, we can easily extend it by assigning~$0$ to each~$z_i$. This construction does not affect our upper bound on the number of variables.
\end{proof}

\section{Kernel Lower Bounds}\label{section:lower}

We will now complete the dichotomy by showing that if~\MOS{\Gamma} is \NP-complete and some~$R \in \Gamma$ is not \mergeable, then the problem admits no polynomial kernelization unless \phcollapse.  The central concept of our lower bound construction is the following definition.

\begin{definition}
A \emph{log-cost selection formula} of arity~$n$ is a formula on variable sets~$X$ and~$Y$, with~$|Y|=n$ and~$|X|=n^{O(1)}$, such that there is no solution where~$Y=0$, but for any~$y_i\in Y$ there is a solution where~$y_i=1$,~$y_j=0$ for~$j \neq i$, and where a fix number~$w_n=O(\log n)$ variables among~$X$ are true. Furthermore, there is no solution where fewer than~$w_n$ variables among~$X$ are true.
\end{definition}

We will show that any~$\Gamma$ as described can be used to construct log-cost selection formulas, and then derive a lower bound from this.  The next lemma describes our constructions.

\begin{lemma} \label{lm:logcostimpl}
The following types of relations can implement log-cost selection formulas of any arity.
\begin{enumerate}
\item A 3-ary relation~$R_3$ such that~$\{(0,0,0), (1,1,0), (1,0,1)\} \subseteq R_3$ and~$(1,0,0) \notin R_3$, together with relations~$(x=1)$ and~$(x=0)$.
\label{lcimpl1}
\item A 5-ary relation~$R_5$ such that~$\{(1,0,1,1,0), (1,0,0,0,0), (0,1,1,0,1), (0,1,0,0,0)\} \subseteq R_5$ and~$(1,0,1,0,0),(0,1,1,0,0) \notin R_5$, together with relations~$(x \neq y)$,~$(x=1)$, and~$(x=0)$.
\label{lcimpl2}
\end{enumerate}
\end{lemma}
\begin{proof}
Let~$Y=\{y_1, \ldots, y_n\}$ be the variables over which a log-cost selection formula is requested.  We will create ``branching trees'' over variables~$x_{i,j}$ for~$0 \leq i \leq \log_2 n$,~$1 \leq j \leq 2^i$, as variants of the composition trees used in~\cite{hfreenokernel}.  
Assume that~$n=2^h$ for some integer~$h$; otherwise pad~$Y$ with variables forced to be false, as assumed to be possible in both constructions.

The first construction is immediate.  Create the variables~$x_{i,j}$ and add a constraint~$(x_{0,1}=1)$.  Further, for all~$i,j$ with~$0 \leq i < h$ and~$1 \leq j \leq 2^i$, add a constraint~$R_3(x_{i,j}, x_{i+1,2j-1}, x_{i+1, 2j})$.
Finally, replace variables~$x_{h,j}$ by~$y_j$.  By the requirements on~$R_3$, for every internal variable~$x_{i,j}$, if~$x_{i,j}=1$ then one of its children ~$x_{i+1,2j-1}$ and~$x_{i+1,2j}$ must be true. Thus by transitivity, some variable on each level of the branching tree must be true, making~$Y=0$ is impossible.  Conversely, for any variable~$y_i$ on the leaf level, there is a solution where exactly the variables along the path from the root node to~$y_i$ are true.  
Thus~$w_n=h=\log_2 n$.

The second construction uses the same principle, but the construction is somewhat more involved.  Create variables~$x_{i,j}$ and a constraint~$(x_{0,1}=1)$ as before.  In addition, introduce for every~$0 \leq i \leq h-1$ two variables~$l_i$,~$r_i$ and a constraint~$(l_i \neq r_i)$.  Now the intention is that~$(l_i,r_i)$ decides whether the path of true variables from the root to a leaf should take a left or a right turn after level~$i$.  
Concretely, add for every~$i,j$ with~$0 \leq i \leq h-1$ and~$1 \leq j \leq 2^i$ a constraint~$R_5(l_i, r_i, x_{i,j}, x_{i+1,2j-1}, x_{i+1, 2j})$.
Now for every true variable~$x_{i,j}$, it is not allowed that~$x_{i+1,2j-1}=x_{i+1,2j}=0$, while depending on~$l_i$ and~$r_i$, either~$(x_{i+1,2j-1}=1, x_{i+1,2j}=0)$ or~$(x_{i+1,2j-1}=0, x_{i+1,2j}=1)$ is allowed.  This rules out the case~$Y=0$, while for each set of values of~$l_i$,~$r_i$ it is allowed to set among variables~$x_{i,j}$ exactly the variables along a path from the root to a leaf~$y_i$ to true, and other variables to false. 
In total, exactly two variables not among~$Y$ are true per level in such an assignment, making~$w_n = 2h = 2\log_2 n$.
\end{proof}

We now reach the technical part, where we show that any relation which is not \mergeable can be used to construct a relation as in Lemma~\ref{lm:logcostimpl}.  The constructions are based on the concept of a \emph{witness} that some relation~$R$ lacks a certain closure property.  For instance, if~$R$ is not \mergeable, then there are four tuples~$\alpha, \beta, \gamma, \delta \in R$ to which the \mergeoperation applies, but such that~$\alpha \land (\beta \lor \gamma) \notin R$; these four tuples form a witness that~$R$ is not \mergeable.
Using the knowledge that such witnesses exist, we use the approach of Schaefer~\cite{schaefersat}, identifying variables according to their occurrence in the tuples of the witness, to build relations with the properties we need.

\begin{lemma} \label{lm:constant}
Let~$\Gamma$ be a set of relations such that \MOS{\Gamma} is \NP-complete and some~$R \in \Gamma$ is not \mergeable.
Under a constraint that at most~$k$ variables are true,~$\Gamma$ can be used to force~$(x=0)$ and~$(x=1)$.  
Furthermore, there is an implementation of~$(x=y)$ using~$R$,~$(x=0)$, and~$(x=1)$.
\end{lemma}

\begin{proof}
First of all, we show how to force~$(x=1)$.  Since \MOS{\Gamma} is \NP-complete, it contains some relation that is not zero-valid; let~$R \in \Gamma$ be such a relation.
If~$R$ is one-valid, then~$R(x,\ldots,x)$ is equivalent to~$(x=1)$. Else, let~$r$ be the arity of~$R$ and let~$I$ be a maximal set such that~$R(x_1, \ldots, x_r)$ holds for~$x_i=1$ for~$i \in I$,~$x_i=0$ else.
Identify all~$x_i$,~$i \in I$, to a single variable~$x$, and all~$x_i$,~$i \notin I$, to a single variable~$y$.  This forms a new constraint~$R'(x,y)$, where~$(1,0) \in R'(x,y)$ and~$(0,0),(1,1) \notin R'(x,y)$.  Thus~$R'$ is either~$(x=1\land y=0)$ or~$(x \neq y)$.
In the former case we are done; in the latter case, constraints~$x \neq y_i$ for~$1 \leq i \leq k+1$ force~$x=1$ and all~$y_i=0$ in any solution with at most~$k$ true variables.

Now we can use this to force~$(x=0)$ and~$(x=y)$.
Let~$\alpha$ through~$\delta$ be a witness that~$R$ is not \mergeable; let~$\sigma=\alpha \land (\beta \lor \gamma) \notin R$ be the produced tuple.
Notice that~$\beta < \sigma < \alpha$, meaning that the positions of~$R$ are of four types: those where~$\beta < \sigma$, those where~$\sigma < \alpha$, and optionally positions which are constant among these tuples, i.e.\ true in~$\beta$ or false in~$\alpha$.
Call these positions~$C_x$,~$C_y$,~$C_1$, and~$C_0$, in the order they were introduced.
Place a variable~$z_1=1$ in all positions~$C_1$, if any, and variables~$x$ and~$y$ in all positions~$C_x$ resp.~$C_y$.
Now, if there are no positions~$C_0$, then this creates a constraint~$R'(x,y)$ such that~$R'(x,y) \land R'(y,x)$ implements~$(x=y)$ directly.
This can be used to force~$x=0$: create~$k$ variables~$y_i$ and let~$x=y_i$ for every~$i$.  In any solution with at most~$k$ true variables, all these variables are false.

Otherwise, if there are positions~$C_0$, then place the variable~$y$ in these positions as well, and apply~$R'(x,y) \land R'(y,x)$ again; the result is either~$(x=y)$ or~$(x=y=0)$.
Finally, placing a variable~$z_0=0$ in positions~$C_0$ lets us implement~$(x=y)$ as above.
\end{proof}

\begin{lemma} \label{lm:gammalogcost}
Let \MOS{\Gamma} be NP-complete, and not \mergeable.  Then \MOS{\Gamma} can express a log-cost selection formula of any arity.
\end{lemma}
\begin{proof}
Let~$R \in \Gamma$ be a relation that is not \mergeable, and 
let~$\alpha$ through~$\delta$ be a witness of this.  By Prop.~\ref{prop:sunflowerdecomp}, partition the positions of~$R$ into core and petals in a way that agrees with the witness.
Group the variables w.r.t. their values in these four tuples into constant variables~$Z_1,Z_0$, non-constant core variables~$C_{10}$ and $C_{01}$, and non-constant petal variables~$P_{11}$, $P_{10}$, $P_{01}$ (where the indices indicate membership in~$\alpha$ and~$\gamma$, as~$\beta$ and~$\delta$ are now determined by this). 
Identify variables according to type, and order them in the order of the previous sentence.  We now have a relation whose arity depends on which variable types that are represented in the witness.  
In the case that all seven types are present, we have implemented a~$7$-ary relation~$R_7$ about which we know the following (the final tuple is produced on the witness tuples by the \mergeoperation).
\begin{eqnarray}
(1,0, 1,0, 1,1,0) &\in& R_7 \nonumber \\
(1,0, 1,0, 0,0,0) &\in& R_7 \nonumber \\
(1,0, 0,1, 1,0,1) &\in& R_7 \nonumber \\
(1,0, 0,1, 0,0,0) &\in& R_7 \nonumber \\
(1,0, 1,0, 1,0,0) &\notin& R_7 \nonumber
\end{eqnarray}
To distinguish the final tuple from the witness tuples, we can observe that variable types~$P_{11}$,~$P_{10}$, and one further non-constant variable type must be represented by the witness.
The constant positions can be ignored by putting variables~$z_1=1$ and~$z_0=0$ in these positions, by Lemma~\ref{lm:constant}. 
Thus we implement a relation of arity between three and five.

First assume that~$R$ is dual Horn.  Then the tuple~$\beta \lor \gamma \in R$, i.e.~$(1,0,1,1, 1,0,1) \in R_7$, and the variable type~$C_{01}$ or~$P_{01}$ must occur. 
Identify~$C_{01}$ and~$P_{01}$ if both occur, and set~$C_{10}=1$ if this type occurs, implementing a 3-ary relation~$R'$ which matches~$R_3$ of Lemma~\ref{lm:logcostimpl}, with the variable types being~$P_{11}$,~$P_{10}$, and~$(P_{01}=C_{01})$ in the order used in Lemma~\ref{lm:logcostimpl}.
Indeed,~$\{\alpha, \beta \lor \gamma, \beta\} \subseteq R$, representing the positive requirement, while there can be no tuple~$(1,0,0) \in R'$, whether~$C_{10}$ occurs or not.
This and Lemma~\ref{lm:constant} fulfills the conditions of Lemma~\ref{lm:logcostimpl}, part 1.

Otherwise~$R$ is not dual Horn, in which case it is not closed under disjunction.  Using a witness for this, we can implement a 2-ary relation~$R_2$ which is either~$(x \neq y)$ or~$(\neg x \lor \neg y)$.   
Likewise, by \NP-completeness we have a relation which is not Horn, which can implement~$(x \neq y)$ or~$(x \lor y)$.
Combining them, we find that we can always implement~$(x \neq y)$, and thus are free to use~$R_5$ of Lemma~\ref{lm:logcostimpl}.
We implement a relation~$R'$ as before, again letting the variable types appear in the order~$(C_{10}, C_{01}, P_{11}, P_{10}, P_{01})$.
We go through the cases of non-empty non-constant variable types, and show that our relation~$R'$ can implement a relation matching~$R_3$ or~$R_5$ of Lemma~\ref{lm:logcostimpl}.
\begin{enumerate}
\item If~$R'$ has arity three, with the third variable type being~$P_{01}$ or~$C_{01}$, then we implement a relation matching~$R_3$ with $\{(1,1,0), (1,0,1), (0,0,0)\} \subseteq R'$ and~$(1,0,0) \notin R'$.

\item If~$R'$ has arity three, and the third type is~$C_{10}$, then we implement a relation~$R'$ with $\{(1,1,1), (1,0,0), (0,1,0), (0,0,0) \subseteq R'$ and~$(1,1,0) \notin R'$.  Use~$R'(v,x,y) \land R'(w,x,z)$ to implement a relation matching~$R_5$.
\label{glccase2}

\item If the core type~$C_{10}$ is not present, then identify~$C_{01}$ with~$P_{01}$.  This implements a relation~$R'$ matching~$R_3$.

\item If the core type~$C_{01}$ is not present, we need two cases.  If~$(0,1,0,0) \notin R'$, then identify~$C_{10}$ with~$P_{10}$ to produce a 3-ary relation matching~$R_3$. Otherwise, force~$P_{01}=0$ to produce a 3-ary relation as in case~\ref{glccase2}.

\item If the petal type~$P_{01}$ is not present, then~$R'(v,w,x,y) \land R'(w,v,x,z)$ implements a relation matching~$R_5$.

\item If all five types are present, then~$R'(v,w,x,y,z) \land R'(w,v,x,z,y)$ implements a relation matching~$R_5$.
\end{enumerate}
Thus in every case, we meet the conditions of part 1 or 2 of Lemma~\ref{lm:logcostimpl}.
\end{proof}

We now show our result, using the tools of~\cite{DBLP:conf/esa/BodlaenderTY09}.  We have the following definition.
Let~$\Q$ and~$\Q'$ be parameterized problems. 
A \emph{polynomial time and parameter transformation} from~$\Q$ to~$\Q'$ is a polynomial-time mapping~$H:\Sigma^*\times\N\to\Sigma^*\times\N:(x,k)\mapsto(x',k')$ such that
$$\forall(x,k)\in\Sigma^*\times\N:((x,k)\in\Q\Leftrightarrow (x',k')\in\Q')\mbox{ and }k'\leq p(k),$$
for some polynomial~$p$.

We will provide a polynomial time and parameter transformation to \MOS{\Gamma} from Exact Hitting Set$(m)$, defined as follows.
\begin{quotation}
\noindent\textbf{Input:} A hypergraph~$\h$ consisting of~$m$ subsets of a universe~$U$ of size~$n$.

\noindent\textbf{Parameter:}~$m$.

\noindent\textbf{Task:} Decide whether there is a set~$S \subset U$ such that~$|E \cap S|=1$ for every~$E \in \h$. 
\end{quotation}
It was shown in~\cite{DBLP:conf/esa/BodlaenderTY09} that polynomial time and parameter transformations preserve polynomial kernelizability; thus our lower bound will follow.
To establish a lower bound for Exact Hitting Set$(m)$,
we need the following notions from~\cite{BodlaenderDFH08,DBLP:conf/esa/BodlaenderTY09}. Let~$\Q$ be a parameterized problem. A \emph{composition algorithm} for~$\Q$ is an algorithm that on input~$(x_1,k),\dots,(x_t,k)\subseteq\Sigma^*\times\N$ uses time polynomial in~$\sum^t_{i=1}|x_i|+k$ and outputs~$(y,k')$ with~$k'$ bounded by a polynomial in~$k$ and such that~$(y,k')\in\Q$ if and only if~$(x_i,k)\in\Q$ for at least one~$i\in\{1,\dots,t\}$. The problem~$\Q$ is then said to be \emph{compositional}.

The \emph{derived classical problem}~$\tilde{\Q}$ of~$\Q$ is defined by~$\tilde{\Q}=\{x\#1^k\mid (x,k)\in\Q\}$, where~$\#\notin\Sigma$ is the blank letter and~$1$ is any letter from~$\Sigma$.

Following Dom et al.~\cite{colorsids}, we next give our equivalence of a \emph{colored} version of the problem,
which we call Exact CSP$(m+n)$, defined as follows.
\begin{quotation}
\noindent\textbf{Input:} A CSP instance with~$n$ variables of arbitrary finite domain, and~$m$ constraints Exactly-One($v_{i_1}=b_{i_1}, \dots, v_{i_r}=b_{i_r}$) of arbitrary arity, where each~$v_i$ is a variable and~$b_i$ a value from the respective variable domain.

\noindent\textbf{Parameter:} $m+n$.

\noindent\textbf{Task:} Decide whether there is an assignment of a value to every variable that satisfies each constraint (i.e.\ for each constraint, exactly one statement~$v=b$ is true). 
\end{quotation}
We show that Exact CSP$(m+n)$ admits no polynomial kernelization; the result follows by a trivial problem reduction from Exact CSP$(m+n)$ to Exact Hitting Set$(m)$.
The proof follows the same lines as the lower bound for Unique Coverage in~\cite[Sec.~4.2]{colorsids}, but the construction is somewhat simplified, and the lower bound somewhat stronger (as Exact Hitting Set$(m)$ is equivalent to a special case of Unique Coverage).\footnote{Dom et al.\ also give a proof for the problem \textsc{Bipartite Perfect Code}, which is the same underlying problem as here, but the parameterization is different.}

\begin{lemma} \label{lm:ehsnokernel}
Exact Hitting Set$(m)$ admits no polynomial kernelization unless \phcollapse.
\end{lemma}
\begin{proof}
First, Exact CSP$(m+n)$ is \NP-complete (even if all domains have cardinality 2, in which case it is the Exact Satisfiability problem). 
Also, the problem can be solved in time $O^*(n^m)$.  Decide for each constraint the identity of the variable which will hit it (but not yet its value).
Assuming that a variable~$v$ is chosen for a particular constraint, for every statement~$(v_i=b_j)$ in the constraint with~$v\neq v_i$, 
remove the value~$j$ from the domain of the variable~$v_i$, and restrict the domain of~$v$ to those values which would hit the constraint.  
Repeat for all constraints, backtracking if necessary; the size of the search tree is at most~$n^m$.  
Thus, we may assume in our composition algorithm that the number of input instances is bounded by~$n^m$ (or else we solve all instances in time polynomial in the total input size).

Assume, then, that there are~$t$ input instances.  Let~$n$ be the maximum number of variables and~$m$ the maximum number of constraints; for simplicity of the argument,
assume that all input instances have the same numbers of variables and constraints (or else do trivial padding with unary-domain variables or trivially true constraints,
such as ``variable 1 has exactly one value'').  Number the variables from~$1$ to~$n$ and the constraints from~$1$ to~$m$ in each input instance.

Now create the composed instance.  First collect all the values of variables numbered~$i$ into the domain of a single variable~$v_i'$,
say with values~$(j_1,j_2)$ signifying ``value~$j_2$ in the domain of input instance~$j_1$''.  Similarly concatenate all constraints numbered~$i$ into a single constraint,
over these new domain values.   Note that values stemming from different instances are different, so that a constraint is hit only once in an intended solution
(where all values come from the same instance).  We finally need to add constraints to ensure that all variables take values stemming from the same input instance.

For this, assign to each input instance a number from~$1$ to~$t$ as its ID, and write this in binary form.  Let~$l=\lceil \log_2 t \rceil = O(m \log n)$ be the number of bit levels needed. 
For each pair of values~$(i,i+1)$,~$1 \leq i < n$, and each bit level~$j$,~$1 \leq j \leq l$, add a testing constraint consisting of all values of~$v_i$ for which
the~$j$:th digit of the ID of the originating instance is~$1$, and all values of~$v_{i+1}$ for which the~$j$:th digit of the ID is~$0$.  
This makes~$O(mn\log n)$ extra edges.  
If variables~$v_i$ and~$v_j$ take values from different input instances, then their IDs will differ in some position, which will lead to one of these testing edges
being hit twice or not at all.  Otherwise, each testing edge is hit exactly once.
By transitivity, this forces all variables in the composed instance to take values from the same input instance.
This completes the compositionality proof, showing that Exact CSP with parameter~$m+n$ admits no polynomial kernelization.

The result for Exact Hitting Set with parameter~$m$ follows by a simple reduction.  Let the vertices of the hitting set be the individual variable values,
create for each variable an edge containing all its values, and retain all constraints as edges.  We get an equivalent instance with~$m+n$ edges.
\end{proof}

We can now show the main result of this section. 

\begin{theorem}  \label{th:lower}
Let~$\Gamma$ be a constraint language which is not \mergeable.  Then \MOS{\Gamma} is either polynomial-time solvable, or does not admit a polynomial kernelization unless \phcollapse. 
\end{theorem}
\begin{proof}
By Theorem~\ref{th:mosnp}, \MOS{\Gamma} is either polynomial-time solvable or \NP-complete; assume that it is \NP-complete.
By Lemma~\ref{lm:constant}
we have both constants and the constraint~$(x=y)$, and by Lemma~\ref{lm:gammalogcost} we can implement log-cost selection formulas.  
It remains only to describe the polynomial time and parameter transformation from Exact Hitting Set$(m)$ to \MOS{\Gamma}.  

Let~$\h$ be a hypergraph.  If~$\h$ contains more than~$2^m$ vertices, then it can be solved in time polynomial in the input length~\cite{BH08}; 
otherwise, we create a formula~$\F$ and fix a weight~$k$ so that~$(\F,k)$ is positive if and only if~$\h$ has an exact hitting set.  
Create one variable~$y_{i,j}$ in~$\F$ for every occurrence of a vertex~$v_i$ in an edge~$E_j$ in~$\h$.
For each edge~$E \in \h$, create a selection formula over the variables representing the occurrences in~$E$. 
Finally, for all pairs of occurrences of each vertex~$v_i$, add constraints~$(y_{i,j}=y_{i,j'})$,
and fix~$k=m+\sum_{E \in \h} w_{|E|}$, where~$w_i$ is the weight of an~$i$-selection formula.
We have an upper bound on the value of~$k$ of~$O(m \log n) = O(m^2)$.  

Now solutions with weight exactly~$k$ correspond to exact hitting sets of~$\h$.
Note that~$k$ is the minimum possible weight of the selection formulas, which is taken if exactly one occurrence in each edge is picked.
By the definition of log-cost selection formulas, any solution where more than one occurrence has been picked (if such a solution is possible at all)
will have a total weight which is larger than this, if the weight of the~$y$-variables is counted as well, and thus such a solution to~$\F$ of weight at most~$k$ is not possible. 

As Exact Hitting Set$(m)$ is \NP-complete, it follows from~\cite{DBLP:conf/esa/BodlaenderTY09} that a polynomial kernelization for \MOS{\Gamma} would imply the same for Exact Hitting Set$(m)$, giving our result. 
\end{proof}

Finally, let us remark that Lemma~\ref{lm:constant} can be adjusted to provide~$(x=1)$ and~$(x=y)$ without the use of repeated variables, and that using standard techniques (see~\cite{hfreenokernel} and Theorem~\ref{th:upperfinal}), we can show that the lower bound still applies under the restriction that constraints contain no repeated variables.
Such a restriction can be useful in showing hardness of other problems, e.g., as in~\cite{hfreenokernel}.

\section{Conclusions} \label{section:conclusion}

We presented a dichotomy for \MOS{\Gamma} for finite sets of relations~$\Gamma$, assuming that the polynomial hierarchy does not collapse. 
The characterization of the dichotomy is a new concept we call \mergeability.  
We showed that \MOS{\Gamma} admits a polynomial kernelization if the problem is in \P or if every relation in~$\Gamma$ is \mergeable, while in every other case no polynomial kernelization is possible unless \phcollapse, in which case the polynomial hierarchy would collapse to the third level.

It might be interesting to compare our kernelization dichotomy to the approximation properties of \MOS{\Gamma}, as characterized by Khanna et al.~\cite{KhannaSTW00}.  The \mergeability property cuts through the classification of Khanna et al.\ as follows (we use the terms from B.4 of~\cite{KhannaSTW00}).
For every~$\Gamma$ such that \MOS{\Gamma} is known to be in APX, \MOS{\Gamma} admits a polynomial kernelization, while no problem identified as being Min Horn Deletion-complete is \mergeable.\footnote{Every \mergeable problem closed under disjunction is IHSB+, i.e., APX-complete.}
The remaining classes are cut through (e.g., among the affine relations, the relation~$(x + y + z = 1 \textrm{ (mod 2)})$ is \mergeable, while~$(x + y + z=0 \textrm{ (mod 2)})$ is not).  We also get kernelizations for some problems where the corresponding SAT problem is \NP-complete, e.g., Min Ones Exact Hitting Set for sets of bounded arity, where no approximation is possible unless \P=\NP.

\section*{Acknowledgements}
We are thankful to Gustav Nordh and D\'aniel Marx for helpful and interesting discussions.

\end{document}